\documentclass[11pt]{article}

\usepackage{amsmath,amssymb,amsthm,color}
\usepackage{a4wide}
\usepackage{algorithm,algorithmic}
\usepackage[top=1in, bottom=1in, left=1in, right=1in]{geometry}

\title{Approximating Matrix $p$-norms}
\author{Aditya Bhaskara \thanks{Center for Computational Intractability, and Department of Computer Science, Princeton University. Supported by NSF CCF 0832797. Email: \textsf{bhaskara@cs.princeton.edu}} \and
Aravindan Vijayaraghavan \thanks{Center for Computational Intractability, and Department of Computer Science, Princeton University. Supported by NSF CCF 0832797. Email: \textsf{aravindv@cs.princeton.edu}}}

\date{}

\definecolor{DSgray}{cmyk}{0,0,0,0.7}

\newcommand{\norm}[2]{\lVert #1\rVert_{#2}}
\newcommand{\pnorm}[1]{\lVert #1\rVert_{p}}
\newcommand{\Mi}[1]{M_{#1}}
\newcommand{\xopt}{x^*}

\newcommand{\xx}{\mathbf{x}}
\newcommand{\yy}{\mathbf{y}}
\newcommand{\doh}[2]{\frac{\partial #1}{\partial #2}}
\newcommand{\grad}{\nabla}
\newcommand{\rr}{\mathbb{R}}
\newcommand{\rplus}{\rr_{+}}
\newcommand{\eps}{\varepsilon}
\newcommand{\ptop}{p \mapsto p}
\newcommand{\qtop}{q \mapsto p}
\newcommand{\Sqn}{\mathcal{S}_q^n}
\newcommand{\num}{\mathcal{N}}
\newcommand{\del}{\delta}
\newcommand{\E}{\mathbb{E}}

\newcommand{\dfi}[1]{\frac{\partial f}{\partial #1}}
\newcommand{\dfii}[1]{\frac{\partial^2 f}{\partial #1 ^2}}
\newcommand{\dfij}[2]{\frac{\partial^2 f}{\partial #1 \partial #2}}

\newtheorem{theorem}{Theorem}[section]
\newtheorem{lemma}[theorem]{Lemma}
\newtheorem{defn}[theorem]{Definition}
\newtheorem{prop}[theorem]{Proposition}

\newtheorem{claim}[theorem]{Claim}
\newtheorem{obser}[theorem]{Observation}

\begin{document}
\maketitle
\begin{abstract}
We consider the problem of computing the $\qtop$ norm of a matrix $A$, which is defined for $p,q \ge 1$, as
\[ \norm{A}{\qtop} = \max_{x \neq \vec{0}} \frac{\norm{Ax}{p}}{\norm{x}{q}}. \]

This is in general a non-convex optimization problem,
and is a natural generalization of the well-studied question of computing singular values (this
corresponds to $p=q=2$). Different settings of parameters give rise to a variety of
known interesting problems (such as the Grothendieck problem when $p=1$ and $q=\infty$).
However, very little is understood about the approximability of the problem for
different values of $p,q$.

Our first result is an efficient algorithm for computing the $\qtop$ norm of matrices with non-negative
entries, when $q \ge p \ge 1$. The algorithm we analyze is based on a natural fixed point iteration, which can
be seen as an analog of power iteration for computing eigenvalues.

We then present an application of our techniques to the problem of constructing a scheme for oblivious routing
in the $\ell_p$ norm. This makes constructive a recent existential result of Englert and R\"acke~\cite{racke} on $O(\log n)$
competitive oblivious routing schemes (which they make constructive only for $p=2$).

On the other hand, when we do not have any restrictions on the entries (such as non-negativity), we prove
that the problem is NP-hard to approximate to any constant factor, for $2<p \le q$ and $p \le q < 2$ (these are precisely the ranges of $p,q$ with $p \le q$ where constant factor approximations are not known). In this range, our techniques also
show that if $\mathsf{NP} \notin \mathsf{DTIME}(n^{\text{polylog}(n)})$, the problem cannot be approximated to a
factor $2^{(\log n)^{1-\eps}}$, for any constant $\eps>0$.
\end{abstract}
\maketitle
\newpage

\section{Introduction}
We study the problem of computing norms of matrices. The $\ell_q$ to
$\ell_p$ norm of a matrix $A \in \rr^{m \times n}$ is defined to be
\[ \norm{A}{\qtop} = \max_{x \in \rr^n, x \neq \vec{0}}
~\frac{\norm{Ax}{p}}{\norm{x}{q}}, \qquad \text{where } \norm{x}{p} = (|x_1|^p + \dots +|x_n|^p)^{1/p} \]

Throughout, we think of $p, q \ge 1$. If we think of the matrix as an operator from $\rr^n$ with the $\ell_q$ norm to the space
$\rr^m$ with $\ell_p$ norm, the norm $\norm{A}{\qtop}$ measures the `maximum stretch' of
a unit vector. 

Computing the $\qtop$-norm of a matrix is a natural optimization
problem. For instance, it can be seen as a natural generalization
of the extensively studied problem of computing the largest singular value of a
matrix~\cite{matrix}. This corresponds to the case $p=q=2$.  When $p=1$ and $q=\infty$, it turns out to be
the well-studied Grothendieck problem~\cite{grothendieck, an},
which is defined as
\[ \max_{x_i, y_i \in \{-1, 1\}} \sum_{i,j} a_{ij} x_i y_j. \]
Thus for different settings of the parameters, the problem seems to have very different flavors.

We study the question of approximating $\norm{A}{\qtop}$ for different
ranges of the parameters $p, q$. The case $p=q$ is referred to as the
matrix $p$-norm (denoted by $\norm{A}{p}$), and has been considered in
the scientific computing community. For instance, it is known to have
connections with matrix condition number estimates (see \cite{higham}
for other applications). Computing $\norm{A}{\qtop}$
 has also been studied because of its connections to robust optimization \cite{daureen}.
Another special case which has been studied~\cite{boyd, daureen} is one where 
the entries of the matrix $A$ are restricted to be non-negative. 
Such instances come up in graph theoretic problems, like in the $\ell_p$ 
oblivious routing question of \cite{racke}.

Note that computing the matrix $\qtop$ norm is a problem of maximizing a convex function 
over a convex domain. While a convex function can be minimized efficiently over convex domains
 using gradient descent based algorithms, it is in general hard to maximize them. 
Thus it is interesting that our algorithm can efficiently compute the norm for non-negative matrices
for a range of parameters.

\paragraph{Known algorithms.} Very little is known about approximating
$\norm{A}{\qtop}$ in general. For computing $p$-norms (i.e., $q=p$), polynomial
time algorithms for arbitrary $A$ are known to exist only
for $p = 1, 2,$ and $\infty$. For the general problem, for $p \le 2$, $q >2$,
Nesterov\cite{nesterov} shows that the problem can be approximated to
a constant factor (which can be shown to be $<2.3$), using a
semidefinite programming relaxation. When the matrix has only
non-negative entries, this relaxation can be shown to be
exact~\cite{daureen}.

For other ranges of $p,q$, the best known bounds are polynomial factor approximations, obtained by
`interpolating'. For instance, for computing $\norm{A}{\ptop}$, computing the vectors
that maximize the norm for $p=1, 2, \infty$, and picking the best of them gives an
$O(n^{1/4})$ approximation for all $p$ (see~\cite{higham}). For the general problem
of computing $\norm{A}{\qtop}$, Steinberg~\cite{daureen} gives an algorithm with an improved
guarantee of $O(n^{25/128})$, by taking into account the approximation algorithms
of Nesterov for certain ranges.

These algorithms use H\"older's inequality, and a fact which follows from the duality of
$\ell_p$ spaces (this sometimes allows one to `move' from one range of parameters to another):
\[ \norm{A}{\qtop} = \norm{A^T}{p' \mapsto q'}, \] where $A^T$ is the
transpose of the matrix $A$, and $p'$ and $q'$ are the `duals' of $p,
q$ respectively (i.e. $1/p + 1/p' = 1$). See Appendix~\ref{sec:app:lpspace-duality}
for a proof.

\paragraph{The hardness front.} The problem is known to be NP-hard in the range
$q \ge p \ge 1$~\cite{daureen}. Very recently in independent work, \cite{hendrix} 
show that it is NP-hard to compute the $p$-norm to
arbitrary relative precision when $p \not\in \{1, 2, \infty\}$ (i.e.,
there cannot be a $(1+\del)$ approximation algorithm with run time
poly$(n,m, 1/\del)$).

\subsection{Our Results}
\paragraph{Non-negative matrices.} We first consider the case of
matrices $A$ with non-negative entries. Here we prove that if $1 \le p
\le q$, then $\norm{A}{\qtop}$ can be computed in polynomial
time. More precisely we give an algorithm which gives a $(1+\del)$
approximation in time polynomial in $n, m$, and $(1/\del)$.

Thus in particular, we give the first poly time guarantee (to the best
of our knowledge) for computing the matrix $p$-norm for non-negative
matrices. We give an analysis of a power iteration type algorithm for
computing $p$-norms proposed by Boyd~\cite{boyd}. The algorithm performs
a fixed point computation, which turns out to mimic
power iteration for eigenvalue computations.

Heuristic approaches to many optimization problems involve finding
solutions via fixed point computations. Our analysis proves polynomial
convergence time for one such natural fixed point algorithm. These
techniques could potentially be useful in other similar settings.
We believe that this algorithm could be useful as an optimization tool
for other problems with objectives that involve $p$-norms (or as a
natural extension of eigenvalue computations). We now mention one such
application, to oblivious routing in the $\ell_p$ norm.

\paragraph{Application to Oblivious Routing.}
In the oblivious routing problem, we are given a graph $G$, and we need to output a `routing scheme', namely a unit flow between every pair of vertices. Now given a set of demands (for a multicommodity flow), we can route them according to this scheme (by scaling the flows we output according to the demand in a natural way), and the total flow on each edge is obtained. The aim is to compete (in terms of, for instance, max congestion) with the best multicommodity flow `in hindsight' (knowing the set of demands). 

For max-congestion (maximum total flow on an edge -- which is the $\ell_\infty$ norm of the vector of flows on edges), a beautiful result of \cite{racke1} gives an $O(\log n)$ competitive routing scheme. Englert and R\"acke~\cite{racke} recently showed that there exists an oblivious routing scheme which attains a competitive ratio of $O(\log n)$ when
the objective function is the $\ell_p$-norm of the flow vector ($|E|$
dimensional vector). However, they can efficiently
compute this oblivious routing scheme only for $p=2$.

From the analysis of our algorithm, we can prove that for matrices with
strictly positive entries there is a unique optimum. Using this and a related idea
(Section~\ref{sec:stability}), we can make the result of \cite{racke}
constructive. Here matrix $p$-norm computation is used as a `separation oracle' in a
multiplicative weights style update, and this gives an $O(\log n)$-competitive oblivious
routing scheme for all $\ell_p$-norms ($p \geq 1$).

\paragraph{Hardness  of approximation.}
For general matrices (with negative entries allowed), we
show the inapproximability of \emph{almost polynomial factor} for
computing the $\qtop$ norm of general matrices when $q \geq p$ and both $p,q$ are $>2$. By duality, this implies the same hardness when both $p,q$ are $<2$ and $q \ge p$.\footnote{When $p \le 2$ and $q\ge 2$, Nesterov's algorithm gives a constant factor approximation.}

More precisely, for these ranges, we prove that computing $\norm{A}{\qtop}$ upto any constant factor is NP-hard.
Under the stronger assumption that $\mathsf{NP} \notin \mathsf{DTIME}(2^{\text{polylog}(n)})$,
we prove that the problem is hard to approximate to a factor of $\Omega(2^{(\log
  n)^{1-\eps}})$, for any constant $\eps>0$.

\noindent{\em Techniques.} We first consider $\ptop$ norm approximation,
for which we show constant factor hardness by a
gadget reduction from the gap version of MaxCut. Then we show that the
$\ptop$ norm multiplies upon tensoring, and thus we get the desired
hardness amplification. While the proof of the small constant hardness
carries over to the $\qtop$ norm case with $q > p >2$, in general
these norms do not multiply under tensoring. We handle this by giving
a way of starting with a hard instance of $\ptop$ norm computation (with additional structure, as will be important),
and convert it to one of $\qtop$ norm computation.

We find the hardness results for computing the $\qtop$ norm interesting because the bounds are very
similar to hardness of combinatorial problems like label cover, and it applies to a natural
numeric optimization problem.

The question of computing $\norm{A}{\infty \mapsto p}$ has a simple alternate formulation (Definition~\ref{defn:lvp}): given vectors $\mathbf{a_1},\mathbf{a_2},\dots,\mathbf{a_n}$, find a $\{\pm\}$ combination of the vectors so as to maximize the length (in $\ell_p$ norm) of the resultant. The previous hardness result also extends to this case. 

\paragraph{Comparison with previous work.} For clarity, let us now tabulate our algorithmic and hardness results
in Tables~\ref{table:old} and show how they compare with known results for different values of the parameters $p,q$. Each row and column in the table gives three things: the best known algorithm in general, the best known approximation algorithm when all entries of $A$ are non-negative, and the best known hardness for this range of $p,q$. An entry saying ``NP-hard'' means that only exact polynomial time algorithms are ruled out.
\begin{table}[h]\label{table:old}
\begin{center}
\caption{\textbf{Previous work}}
\vspace{6pt}
\begin{tabular}{|c|c||c|c|c|}
\hline
           & & $1<q<2$ & $q=2$ & $q>2$ \\
\hline
\hline
Best Approximation   &$1<p<2$& poly(n)& $O(1)$ \cite{nesterov}&
$O(1)$\cite{nesterov}\\
Hardness        &&$p < q$: NP-hard & NP-hard & NP-hard\\
Non-negative&&&Exact \cite{daureen}& Exact \cite{daureen} \\
matrices &&&&\\
\hline
\hline
Best Approximation   &$p=2$ & poly(n)& Exact & $O(1)$ \cite{nesterov} \\
 Hardness&&&&NP-hard\\
Non-negative&&&Exact&Exact \cite{daureen} \\
matrices &&&&\\
\hline
\hline
Best Approximation   &$p>2$ & poly($n$)& poly($n$)& poly($n$)\\
 Hardness&&&&$p<q$: NP-hard\\
Non-negative&&&& \\
matrices&&&&\\
\hline
\hline
\end{tabular}
\vspace{10pt}
\caption{\textbf{Our results.} We give better algorithms for
non-negative matrices and obtain almost-polynomial hardness results
when $q \geq p$.}
\vspace{6pt}
\begin{tabular}{|c|c||c|c|c|}
\hline
           & & $1<q<2$ & $q=2$ & $q>2$ \\
\hline
\hline
Hardness        &$1<p<2$&$p \leq q$: $2^{(\log n)^{1-\epsilon}}$-hard&&\\
Non-negative matrices&&$p \leq q$: Exact&Exact& Exact\\
\hline
\hline
 Hardness&$p=2$&&&\\
Non-negative matrices&&&Exact&Exact\\
\hline
\hline
 Hardness& $p>2$&&&$p\leq q$: $2^{(\log n)^{1-\epsilon}}$-hard\\
Non-negative matrices&&&&$p \leq q$: Exact \\
\hline
\hline
\end{tabular}
\end{center}
\end{table}
%\vspace{-.4in}
\paragraph{Discussion and open questions.} All our algorithms and hardness results apply to the case $p \le q$, but we do not know either of these (even for non-negative matrices) for $p>q$ (which is rather surprising!). For algorithmic results (for positive matrices, say) the fact that we can optimize $\norm{Ax}{p}/\norm{x}{q}$ seems closely tied to the fact that the set $\{x~:~\norm{x}{p}/\norm{x}{q} > \tau\}$ is convex for any $\tau>0$ and $p \le q$. However, we are not aware of any formal connection. Besides, when $p>q$, even for non-negative matrices there could be multiple optima (we prove uniqueness of optimum when $p \le q$).

On the hardness front, the $q<p$ case seems more related to questions like the Densest $k$-subgraph problem (informally, when the matrix is positive and $p<q$, if there is a `dense enough' submatrix, the optimum vector would have most of its support corresponding to this). Thus the difficulties in proving hardness for the norm question may be related to proving hardness for densest subgraph.

Hypercontractive norms (corresponding to $q<p$) have been well-studied~\cite{kkl}, and have also found prior use in inapproximability results for problems like maxcut. Also, known integrality gap instances for unique games~\cite{kv} are graphs that are hypercontractive. We believe that computability of hypercontractive norms of a matrix could reveal insights into the approximability of problems like small set expansion~\cite{sse} and the planted dense $k$-subgraph problem~\cite{dks}.

\subsection{Related work.}
A question that is very related to matrix norm computation is the $L_p$ Grothendieck problem, which has been studied earlier by~\cite{naor}. The problem is to compute 
\[\max_{||\mathbf{x}||_{p}\leq 1} \mathbf{x}^t B \mathbf{x} \]
The question of computing $\norm{A}{p \mapsto 2}$ is a special case of the $L_p$ Grothendieck problem (where $B \succeq 0$).
\cite{naor} give an optimal (assuming UGC) $O(p)$ approximation algorithm. For $B$ being p.s.d., constant factor approximation algorithms are known, due to \cite{nesterov}. Computing $\norm{A}{\infty \mapsto 2}$ reduces to maximizing a quadratic form over ${\pm 1}$ domain for p.s.d matrices \cite{cw,nesterov}. 

Recently, \cite{madhur} studies an optimization problem which has an $\ell_p$ norm objective -- they wish to find the best $k$-dimensional subspace approximation to a set of points, where one wishes to minimize the $\ell_p$ distances to the subspace (there are other problems in approximation theory which are of similar nature). When $k=n-1$ this can be shown to reduce to the $L_p$ Grothendieck problem for the matrix $A^{-1}$.

\subsection{Paper Outline} \label{sec:intro:outline}

We start by presenting the algorithm for positive matrices (Section~\ref{sec:algo:description}), and prove
poly time convergence (Section~\ref{sec:algo:analysis}). Some additional properties of the optimization problem are discussed in Section~\ref{sec:stability} (such as unique maximum, concavity around optimum),
which will be useful for an oblivious routing application. This will be presented in Section~\ref{sec:oblivious}.
Finally in Section~\ref{sec:lower-bounds}, we study the inapproximability of the problem: we first show a constant factor hardness for $\norm{A}{\ptop}$ (Section~\ref{sec:lb:ptop}), and show how to amplify it (Section~\ref{sec:tensoring}). Then we use this to show hardness for $\norm{A}{\qtop}$ in section~\ref{sec:lb:ptoq}.

\section{Notation and Simplifications}\label{sec:intro:notation}

We write $\rplus$ for the set of non-negative reals. For a matrix $A$,
we let $A_i$ denote the $i$th row of $A$. Also $a_{ij}$ denotes the
element in the $i$th row and $j$th column. Similarly for a vector
$\xx$, we denote the $i$th co-ordinate by $x_i$. We say that a vector
$\xx$ is {\em positive} if the entries $x_i$ are all $>0$. Finally,
for two vectors $\xx,\yy$, we write $\xx \propto \yy$ to mean that
$\xx$ is {\em proportional to} $\yy$, i.e., $\xx = \lambda \yy$ for
some $\lambda$ (in all places we use it, $\lambda$ will be $>0$).

For our algorithmic results, it will be much more convenient to work with matrices where we restrict the entries to be in $[1/N,1]$, for some parameter $N$ (zero entries can cause minor problems). If we are interested in a $(1+\del)$ approximation, we can first scale $A$ such that the largest entry is $1$, pick $N \approx (m+n)^2/\delta$, where $m, n$ are the dimensions of the matrix, and work with the matrix $A+\frac{1}{N} J$ (here $J$ is the $m\times n$ matrix of ones). The justification for this can be found in Appendix~\ref{sec:app:simplification}. We will refer to such $A$ as a {\em positive matrix}.

\section{An Iterative Algorithm}\label{sec:algo}
In this section, we consider positive matrices $A$, and prove that if $1 < p \le q$, we can efficiently compute $\norm{A}{\qtop}$. Suppose $A$ is of dimensions $n \times n$, and define $f: \rr^n \mapsto \rr$ by
\[ f(\xx) = \frac{ \norm{A \xx }{p}}{\norm{\xx}{q}} = \frac{ (\sum_i |A_i \xx|^p )^{1/p}}{(\sum_i |x_i|^q )^{1/q}}. \]

We present an algorithm due to Boyd~\cite{boyd}, and prove that it converges quickly to the optimum vector. The idea is to
consider $\grad f$, and rewrite the equation $\grad f=0$ as a fixed
point equation (i.e., as $S \xx = \xx$, for an appropriate operator
$S$). The iterative algorithm then starts with some vector $\xx$, and
applies $S$ repeatedly. Note that in the case $p=2$, this mimics the familiar power
iteration (in this case $S$ will turn out to be multiplication by the
matrix $A$ (up to normalization)).

\subsection{Algorithm description}\label{sec:algo:description}
Let us start by looking at $\grad f$.
\begin{equation}
\doh{f}{x_i} = \frac{\norm{x}{q} \norm{Ax}{p}^{1-p} \cdot \sum_j a_{ij} |A_j x|^{p-1}  - \norm{Ax}{p} \norm{x}{q}^{1-q} \cdot |x_i|^{q-1}}{ \norm{x}{q}^{2}}
\end{equation}
At a critical point, $\doh{f}{x_i} = 0$ for all $i$. Thus for all $i$,
\begin{equation}\label{eq:critical}
|x_i|^{q-1} = \frac{\norm{x}{q}^q}{\norm{Ax}{p}^p} \cdot \sum_j a_{ij} |A_j x|^{p-1}
\end{equation}

Define an operator $S: \rplus^n \rightarrow \rplus^n$, with the $i$th
co-ordinate of $Sx$ being (note that all terms involved are positive)
\[ (Sx)_i = \big( \sum_j a_{ij} (A_j x)^{p-1} \big)^{1/(q-1)} \]
Thus, at a critical point, $Sx \propto x$. Now consider the the
following algorithm:\\
({\bf Input.} An $n \times n$ matrix $A$ with all entries in
$[\frac{1}{N}, 1]$, error parameter $\delta$.)
\begin{algorithmic}[1]
\STATE Initialize $x = \frac{\mathbf{1}}{\norm{\mathbf{1}}{p}}$.
\LOOP[$T$ times (it will turn out $T = (Nn) \cdot \textrm{polylog}(N,n,1/\delta)$)]
\STATE set $x \leftarrow Sx$.
\STATE normalize $x$ to make $\norm{x}{q} = 1$.
\ENDLOOP
\end{algorithmic}
A \emph{fixed point} of the iteration is a vector $x$ such that $Sx
\propto x$. Thus every critical point of $f$ is a fixed point. It turns out
that every {\em positive} fixed point is also a critical
point. Further, there will be a unique positive fixed point, which is
also the unique maximum of $f$.

\subsection{Analyzing the Algorithm} \label{sec:algo:analysis}

We will treat $f(\xx)$ as defined over the domain $\rplus^n$. Since the
matrix $A$ is positive, the maximum must be attained in $\rplus^n$.
Since $f$ is invariant under scaling $x$, we restrict our attention to
points in $\Sqn = \{x~:~x \in \rplus^n, \norm{x}{q}=1 \}$. Thus the
algorithm starts with a point in $\Sqn$, and in each iteration moves
to another point, until it converges.

First, we prove that the maximum of $f$ over $\rplus^n$ occurs at an interior point (i.e., none of the co-ordinates are zero). Let $\xopt$ denote a point at which maximum is attained, i.e., $f(\xopt) = \norm{A}{\qtop}$ ($\xopt$ need not be unique). Since it is an interior point, $\grad f =0$ at $\xopt$, and so $\xopt$ is a fixed point for the iteration.

\begin{lemma}\label{lem:nonzero-coords}
Let $\xopt \in \Sqn$ be a point at which $f$ attains maximum. Then each co-ordinate of $x^*$ is at least $\frac{1}{(Nn)^2}$.
\end{lemma}

The proof of this can be found in Section~\ref{sec:stability}. Next, we show that with each iteration, the value of the function cannot decrease. This was proved by~\cite{boyd} (we refer to their paper for the proof).

\begin{lemma}(\cite{boyd})\label{lem:monotone}
For any vector $x$, we have
\[ \frac{\norm{ASx}{p}}{\norm{Sx}{q}} \ge \frac{\norm{Ax}{p}}{\norm{x}{q}} \]
\end{lemma}

The analysis of the algorithm proceeds by maintaining two {\em potentials}, defined by
\[ m(x) = \min_i \frac{(Sx)_i}{x_i} \qquad \textrm{ and } \qquad M(x) = 
\max_i \frac{(Sx)_i}{x_i}.\]
If $x$ is a fixed point, then $m(x) = M(x)$. Also, from Section~\ref{sec:algo:description}, each is equal to $\big( \frac{\norm{x}{q}^q}{\norm{Ax}{p}^p} \big)^{1/(q-1)}$. As observed in \cite{boyd}, these quantites can be used to `sandwich' the norm -- in particular,

\begin{lemma}\label{lem:sandwich}
For any positive vector $x$ with $\norm{x}{q}=1$, we have
\[ m(x)^{q-1} \leq \norm{A}{\qtop}^p \leq M(x)^{q-1}\]
\end{lemma}

The lemma is crucial -- it relates the norm (which we wish to compute) to certain quantities we can compute starting with any positive vector $x$. We now give a proof of this lemma. Our proof, however, has the additional advantage that it immediately implies the following:
\begin{lemma}\label{lem:uniquemax}
The maximum of $f$ on $\Sqn$ is attained at a unique point $\xopt$. Further, this $\xopt$ is the unique critical point of $f$ on $\Sqn$ (which also means it is the unique fixed point for the iteration).
\end{lemma} 
\begin{proof}[Proof (of Lemma~\ref{lem:sandwich})]
Let $x \in \Sqn$ be a positive vector. Let $\xopt \in \Sqn$ be a vector which maximizes $f(x)$.

The first inequality is a simple averaging argument: 
\begin{align}
\frac{\sum_i x_i \cdot (Sx)_i^{q-1}}{\sum_i x_i \cdot x_i^{q-1}} &= 
   \frac{\sum_i x_i \cdot \sum_j a_{ij} (A_j x)^{p-1}}{\sum_i x_i^q}\\
   &= \frac{\sum_j (A_j x)^{p-1} \sum_i a_{ij} x_i}{\sum_i x_i^q}\\
   &= \frac{\sum_j (A_j x)^p}{\sum_i x_i^q} = \frac{\norm{Ax}{p}^p}{\norm{x}{q}^q} \le \norm{A}{\qtop}^p \label{eq:sandwich:avg}
\end{align}
The last inequality uses $\norm{x}{q}=1$. Thus there exists an index $i$ such that $(Sx)_i^{q-1}/x_i^{q-1} \le \norm{Ax}{\qtop}$.

The latter inequality is more tricky -- it gives an upper bound on $f(\xopt)$, no matter which $x \in \Sqn$ we start with. To prove this, we start by observing that $\xopt$ is a fixed point, and thus for all $k$,
\[ m(\xopt)^{q-1} = M(\xopt)^{q-1} = \frac{(S\xopt)_k^{q-1}}{(\xopt)_k^{q-1}}. \]

Call this quantity $\lambda$. Now, let $\theta >0$ be the smallest real number such that $x - \theta \xopt$ has a zero co-ordinate, i.e., $x_k = \theta \xopt_k$, and $x_j \geq
\theta \xopt_j$ for $j \neq k$. Since $\norm{x}{q} = \norm{\xopt}{q}$ and $x \neq \xopt$, $\theta$ is well-defined, and $x_j > \theta \xopt_j$ (strictly) for some index $j$. Because of these, and since each $a_{ij}$ is strictly positive, we have $Sx > S(\theta \xopt) = \theta^{(p-1)/(q-1)} S(\xopt)$ (clear from the definition of $S$).

Now, for the index $k$, we have
\[ \frac{ (Sx)_k^{q-1}}{x_k^{q-1}} > \frac{\theta^{p-1} (S\xopt)_k^{q-1}}{(\theta \xopt_k)^{q-1}} = \theta^{(p-q)} \cdot \lambda \]
Thus we have $M(x)^{q-1} > \lambda$ (since $q \geq p$, and $0<\theta<1$), which is what we wanted to prove.
\end{proof}

Let us see how this implies Lemma~\ref{lem:uniquemax}.
\begin{proof}[Proof (of Lemma~\ref{lem:uniquemax})]
Let $\xopt \in \Sqn$ denote a vector which maximizes $f$ over $\Sqn$ (thus $\xopt$ is one fixed point of $S$). Suppose, if possible, that $y$ is another fixed point. By the calculation in Eq.\eqref{eq:sandwich:avg} (and since $y$ is a fixed point and $\xopt$ maximizes $f$), we have
\[ M(y)^{q-1} = m(y)^{q-1} = \frac{\norm{Ay}{p}^p}{\norm{y}{q}^q} = f(y)^q \le f(\xopt)^{p} \]
Now since $y \neq \xopt$, the argument above (of considering the smallest $\theta$ such that $y - \theta \xopt$ has a zero co-ordinate, and so on) will imply that $M(y)^{q-1} > \lambda = f(\xopt)^p$, which is a contradiction.

This proves that there is no other fixed point.
\end{proof}

The next few lemmas say that as the algorithm proceeds, the value of $m(x)$
increases, while $M(x)$ decreases. Further, it turns out we can quantify how much they
change: if we start with an $x$ such that $M(x)/m(x)$ is `large', the ratio drops significantly in one iteration. 

\begin{lemma}\label{lem:m-increases}
Let $x$ be a positive vector. Then $m(x) \leq m(Sx)$, and $M(x) \geq
M(Sx)$.
\end{lemma}
\begin{proof}
Suppose $m(x) = \lambda$. So for every $i$, we have $(Sx)_i \geq
\lambda x_i$. Now fix some index $i$ and consider the quantity
\[ (SSx)_i^{q-1} = \sum_j a_{ij} (A_j Sx)^{q-1}. \]
Since $A$ is a positive matrix and $(Sx)_i \geq \lambda x_i$, we must
have $(A_j Sx) \geq \lambda \cdot (A_j x)$ for every $j$. Thus
\[ (SSx)_i^{q-1} \geq \lambda^{q-1} \sum_j a_{ij} (A_j x)^{q-1} =
\lambda^{q-1} (Sx)_i^{q-1}. \] This shows that $m(Sx) \geq \lambda$. A
similar argument shows that $M(Sx) \leq M(x)$.
\end{proof}

\begin{lemma}\label{lem:potential:bound}
Let $\xx$ be a positive vector with $\norm{\xx}{q}=1$, and suppose $M(\xx)
\geq (1+\alpha) m(\xx)$. Then $m(S\xx) \geq \big(1+\frac{\alpha}{Nn} \big)
m(\xx)$.
\end{lemma}
\begin{proof}
Let $m(\xx) = \lambda$, and suppose $k$ is an index such that $(S\xx)_k
\geq (1+\alpha)\lambda \cdot x_k$ (such an index exists because $M(\xx) > (1+\alpha) \lambda$. In particular, $(S\xx) \ge \lambda \xx + \alpha \lambda \cdot \mathbf{e}_k$, where $\mathbf{e}_k$ is the standard basis vector with the $k$th entry non-zero. Thus we can say that for every $j$,
\[ A_j (S\xx) \geq \lambda A_j \xx + \alpha \lambda A_j \mathbf{e}_k. \]

The second term will allow us to quantify the improvement in $m(\xx)$. Note that $A_j \mathbf{e}_k = a_{jk} \geq \frac{1}{Nn} A_j \mathbf{1} $
(since $A_{jk}$ is not too small). Now $\mathbf{1} \ge \xx$ since $\xx$ has $q$-norm $1$, and thus we have
\[ A_j (S\xx) \geq \big( 1 + \frac{\alpha}{Nn} \big) \lambda \cdot A_j \xx \]

Thus $(SS \xx)_i^{q-1} \geq \big(1+\frac{\alpha}{Nn}\big)^{q-1}
\lambda^{q-1} (S \xx)_i^{q-1}$, implying that $m(S\xx) \geq
\big(1+\frac{\alpha}{Nn} \big) \lambda$.
\end{proof}

This immediately implies that the value $\norm{A}{\qtop}$ can be computed quickly. In particular,
\begin{theorem}\label{thm:runtime:bound}
For any $\delta > 0$, after $O(Nn \cdot \mathrm{polylog}(N,n,
\frac{1}{\delta}))$ iterations, the algorithm of
Section~\ref{sec:algo:description} finds a vector $x$ such that $f(x) \geq
(1-\delta) f(\xopt)$
\end{theorem}
\begin{proof}
To start with, the ratio $\frac{M(x)}{m(x)}$ is at most $Nn$ (since we start with $\mathbf{1}$, and the entries of the matrix lie in $[1/N,1]$).
Lemma~\ref{lem:potential:bound} now implies that the ratio drops from $(1+\alpha)$ to $(1+\frac{\alpha}{2})$ in $Nn$
iterations. Thus in $T = (Nn) \textrm{polylog}(N,n, 1/\delta)$ steps,
the $x$ we end up with has $\frac{M(x)}{m(x)}$ at most $\big(1 +
\frac{\delta}{(Nn)^{c}} \big)$ for any constant $c$.  This then implies that $f(x) \geq f(\xopt)\big( 1 - \frac{\delta}{(Nn)^c}\big)$, after $T$ iterations.
\end{proof}

\section{Proximity to the optimum} \label{sec:stability}
The argument above showed that the algorithm finds a point $x$ such
that $f(x)$ is close to $f(\xopt)$. We proved that for positive matrices, $\xopt$ is unique, and thus it is natural to ask if the vector we obtain is `close' to $\xopt$. This in fact turns out to be important in an application to oblivious routing which we consider in Section~\ref{sec:oblivious}.

We can prove that the $x$ we obtain after $T = (Nn) \text{polylog}(N, n, 1/\delta)$ iterations is `close' to $\xopt$. The rough outline of the proof is the following: we first show that $f(\xx)$ is strictly concave `around' the optimum \footnote{Note that the function $f$ is not concave everywhere (see Appendix \ref{sec:app:non-convex})}. Then we show that the `level sets' of $f$ are `connected' (precise definitions follow). Then we use these to prove that if $f(x)$ is close to $f(\xopt)$, then $x-\xopt$ is `small' (the choice of norm does not matter much).

Some of these results are of independent interest, and shed light into why the $\qtop$ problem may be easier to solve when $p \le q$ (even for non-negative matrices).

\paragraph{Concavity around the optimum.}
We now show that the neighborhood of every critical point (where $\grad f$ vanishes) is strictly concave. This is another way of proving that every critical point is a maximum (this was the way \cite{racke} prove this fact in the $p=q$ case).

Taking partial derivatives of $f(x)=\frac{\norm{Ax}{p}}{\norm{x}{q}}$, we observe that
\begin{equation} \label{eq:partial1}
 \dfi{x_i}= f(x) \Big( \frac{\sum_k  (A_k x)^{p-1} a_{ki}}{\norm{Ax}{p}^p} - \frac{x_i ^{q-1}}{\norm{x}{q}^q} \Big)
\end{equation}
where $A_k$ refers to the $k^{th}$ row of matrix $A$. 
Now, consider a critical point $z$, with $\norm{z}{q}=1$ (w.l.o.g.). We can also always assume that w.l.o.g. the matrix $A$ is such that $\norm{Az}{p}=1$. Thus at a critical point $z$, as in Eq.\eqref{eq:critical}, we have that for all $i$:

\begin{equation} \label{eq:critical:2}
\sum_k (A_k z)^{p-1} a_{ki}=z_i^{q-1}
\end{equation} 

Computing the second derivative of $f$ at $z$, and simplifying using $\norm{Az}{p} = \norm{z}{q} = 1$, we obtain

\begin{eqnarray}
\left.\frac{1}{p} \cdot \dfij{x_i}{x_j} \right|_{z} &=& (p-1) \sum_k (A_k z)^{p-2} a_{ki}a_{kj} + (q-p) z_i^{q-1} z_j^{q-1}  \label{eq:f11}\\
\left.\frac{1}{p} \cdot \dfii{x_i} \right|_{z} &=& (p-1) \sum_k (A_k z)^{p-2} a_{ki}^2 + (q-p) z_i^{2q-2}  - (q-1)z_i^{q-2} \label{eq:f2} 
\end{eqnarray}

We will now show that the Hessian $H_f$ is negative semi-definite, which proves that $f$ is strictly concave at the critical point $z$. Let $\eps$ be any vector in $\rr^n$. Then we have (the $(q-1) z_i^{q-2}$ in \eqref{eq:f2} is split as $(p-1) z_i^{q-2} +(q-p) z_i^{q-2}$, and $\sum_{i,j}$ includes the case $i=j$)
\begin{align*}
\eps^T H_f \eps &= p(p-1) \Big( \sum_{i,j} \sum_k (A_k z)^{p-2} \cdot
	a_{ki} a_{kj} \cdot \eps_i \eps_j -\sum_i z_i^{q-2} \eps_i ^2 \Big) \\
& ~~~+p(q-p)\Big( \sum_{i,j} (z_i z_j)^{q-1} \eps_i \eps_j - \sum_i z_i^{q-2} \eps_i ^2 \Big)\\
	&\equiv T_1+ T_2 \qquad \text{(say)}
\end{align*}
We consider $T_1$ and $T_2$ individually and prove that they are negative. First consider $T_2$. Since $\sum_i z_i^q=1$, we can consider $z_i^q$ to be a probability distribution on integers $1, \dots, n$. Cauchy-Schwartz now implies that $\E_i [(\eps_i/z_i)^2] \geq \big(\E_i[(\eps_i/z_i)] \big)^2$. This is equivalent to
\begin{equation}
\sum_i z_i^q \cdot \frac{\eps_i^2}{z_i^2} \ge \Big( \sum_i z_i^q \cdot \frac{\eps_i}{z_i} \Big)^2 = \sum_{i,j} z_i^q z_j^q \cdot \frac{\eps_i \eps_j }{z_i z_j}
\end{equation}
Noting that $q \ge p$, we can conclude that $T_2 \le 0$. Now consider $T_1$.
Since $z$ is a fixed point, it satisfies Eq.~\eqref{eq:critical:2}, thus we can substitute for $x_i^{q-1}$ in the second term of $T_1$. Expanding out $(A_k z)$ once and simplifying, we get
\begin{align*}
\frac{T_1}{p(p-1)} &= \sum_{i,j} \sum_k (A_k z)^{p-2} a_{ki}a_{kj} \Big( \eps_i \eps_j -   \frac{z_j}{z_i} \cdot \eps_i ^2 \Big)\\
&= -\sum_k (A_k z)^{p-2} \sum_{i,j} a_{ki}a_{kj} \cdot z_i z_j \cdot \Big( \frac{\eps_i}{z_i}-\frac{\eps_j}{z_j} \Big)^2 \\
&\leq 0
\end{align*} 
This proves that $f$ is concave around any critical point $z$.

\paragraph{Level sets of $f$.}
Let $\Sqn$, as earlier, denote the (closed, compact) set $\{ \xx \in \rplus^n ~:~\norm{\xx}{q}=1 \}$. Let $\mathcal{N}_{\tau}$ denote $\{ x
\in \Sqn~:~f(x) \ge \tau\}$, i.e., $\mathcal{N}_\tau$ is an `upper level set'. (it is easy to see that since $f$ is continuous and $A$ is positive, $\mathcal{N}_\tau$ is closed).

Let $S \subseteq \Sqn$. We say that two points $x$ and $y$ are {\em
  connected} in $S$, if there exists a path (a continuous curve)
connecting $x$ and $y$, entirely contained in $S$ (and this is clearly an equivalence relation). We say that a set
$S$ is \emph{connected} if every $x,y \in S$ are connected in $S$.
Thus any subset of $\Sqn$ can be divided into connected components.
With this notation, we show (\cite{racke} proves the result when $p=q$).
\begin{lemma}\label{lem:racke-connected}
The set $\mathcal{N}_{\tau}$ is connected for every $\tau >0$.
\end{lemma}
This follows easily from techniques we developed so far.
\begin{proof}
Suppose if possible, that $\num_\tau$ has two disconnected components $S_1$ and $S_2$. Since there is a unique global optimum $\xopt$, we may suppose $S_1$ does not contain $\xopt$. Let $y$ be the point in $S_1$ which attains maximum (of $f$) over $S_1$ ($y$ is well defined since $\num$ is closed). Now if $\grad f|_y = \vec{0}$, we get a contradiction since $f$ has a unique critical point, namely $\xopt$ (Lemma~\ref{lem:uniquemax}). If $\grad f|_y \neq \vec{0}$, it has to be normal to the surface $\Sqn$ (else it cannot be that $y$ attains maximum in the connected component $S_1$). Let $\mathbf{z}$ be the direction of the (outward) normal to $\Sqn$ at the point $y$. Clearly, $\langle \mathbf{z}, y \rangle >0$ (intuitively this is clear; it is also easy to check).

We argued that $\grad f|_y$ must be parallel to $\mathbf{z}$, and thus it has a non-zero component along $y$ -- in particular if we scale $y$ (equivalent to moving along $y$), the value of $f$ changes, which is clearly false! Thus $\num_\tau$ has only one connected component.
\end{proof}

Since we need it for what follows, let us now prove Lemma~\ref{lem:nonzero-coords}.
\begin{proof}[Proof of Lemma~\ref{lem:nonzero-coords}]
Let $\xopt$ be the optimum vector, and suppose $\norm{\xopt}{q}=1$. Consider the quantity
\[ f(\xopt)^p = \frac{ \sum_i (A_i \xopt)^p}{\big( \sum_i (\xopt)^q \big)^{p/q}}. \]
First, note that $\xopt_i \neq 0$ for any $i$. Suppose there is such an $i$. If we set $x_i = \delta$, each term in the numerator above increases by at least
$\frac{p \cdot \delta}{N^p}$ (because $A_i x^*$ is at least
$\frac{1}{N}$, and $(\frac{1}{N} + \frac{\delta}{N})^p >
\frac{1}{N^p}(1 + p\delta)$), while the denominator increases from $1$ to $(1+\delta^q)^{p/q} \approx 1+(p/q)\delta^q$ for small enough $\delta$. Thus since $q>1$, we can set $\delta$ small enough and increase the objective.
This implies that $\xopt$ is a positive vector.

\newcommand{\ones}{\vec{\mathbf{1}}}

Note that $A_j \xopt \geq \frac{\ones}{N}\cdot \xopt \geq \frac{1}{N}$
(because the $\norm{\xopt}{1} \ge \norm{\xopt}{q} = 1$). Thus for every $i$,
\[ (S\xopt)_i^{q-1}  = \sum_j a_{ij}(A_j \xopt)^{p-1} \geq \frac{n}{N^p}. \]
Further, $\norm{A}{p}^p \leq n^{p+1}$, because each $a_{ij} \leq 1$ and so
$A_j x \leq n x_{\textrm{max}}$ (where $x_{\textrm{max}}$ denotes
the largest co-ordinate of $x$). Now since Eqn.\eqref{eq:critical}
holds for $\xopt$, we have
\[ n^{p+1} \geq \norm{A}{p}^p = \frac{(S\xopt)_i^{q-1}}{(\xopt)_i^{q-1}} \geq
	\frac{n}{N^p (\xopt)_i^{q-1}}. \]
This implies that $\xopt_i > \frac{1}{(Nn)^2}$, proving the lemma
(we needed to use $q\ge p>1$ to simplify).
\end{proof}

We now show that if $x \in \Sqn$ is `far' from $\xopt$, then $f(x)$
is bounded away from $f(\xopt)$. This, along with the fact that
$\mathcal{N}_\tau$ is connected for all $\tau$, implies that if $f(x)$
is very close to $f(\xopt)$, then $\norm{x-\xopt}{1}$ must be small. For ease of calculation, we give the formal proof only for $p=q$ (this is also the case which is used in the oblivious routing application). It should be clear that as long as we have that the Hessian at $\xopt$ is negative semidefinite, and third derivatives are bounded, the proof goes through.
% BEGIN PASTE
\begin{lemma}[Stability] \label{lem:stability}
Suppose $x \in \Sqn$, with $\norm{x-\xopt}{1} = \delta \leq
\frac{1}{(Nn)^{12}}$. Then
\begin{equation} \label{eq:stability}
f(x) \leq f(\xopt) \Big( 1-\frac{\delta^2}{(Nn)^6} \Big) %some fn of \delta -- I'm not sure
\end{equation}
\end{lemma}
\begin{proof}
Let $\eps$ denote the `error vector' $\eps = x-\xopt$. We will
use the Taylor expansion of $f$ around $\xopt$. $H_f$ denotes
the Hessian of $f$ and $g_f$ is a term involving the third
derivatives, which we will get to later.
 Thus we have: (note that $\grad f$ and $H_f$ are
evaluated at $\xopt$)
\begin{equation}\label{eq:taylor}
f(x)=f(\xopt)+\eps \cdot \grad f_{|\xopt} + \frac{1}{2}~ \eps^T
	H_{f|\xopt} \eps + g_f(\eps')
\end{equation} 
At $\xopt$, the $\grad f$ term is $0$. From the proof above that the Hessian is negative semidefinite, we have
\begin{equation}\label{eqn:hessian}
\eps^T H_f \eps = -p(p-1) \sum_s (A_{s} \xopt)^{p-2} \Big(
	\sum_{i,j} a_{si}a_{sj}\xopt_{i}\xopt_{j} \big(
	\frac{\eps_i}{\xopt_i}-\frac{\eps_j}{\xopt_j}	\big)^2 \Big)
\end{equation}

We want to say that if $\norm{\eps}{1}$ is large enough, this
quantity is sufficiently negative. We should crucially use the
fact that $\norm{\xopt}{p} = \norm{\xopt + \eps}{p} = 1$ (since
$x$ is a unit vector in $p$-norm). This is the same as
\[\sum_i |\xopt_i+\eps_i|^p = \sum_i |\xopt_i|^p. \]
Thus not all $\eps_i$ are of the same sign. Now since
$\norm{\eps}{1} > \delta$, at least one of the $\eps_i$ must
have absolute value at least $\delta/n$, and some other
$\eps_j$ must have the opposite sign, by the above
observation. Now consider the terms corresponding to these
$i, j$ in Eqn.\eqref{eqn:hessian}. This gives
\begin{align}
\eps^T H_f \eps &\leq -p(p-1) \sum_s (A_s \xopt)^{p-2} \cdot
	a_{si} a_{sj} \cdot \frac{\xopt_j}{\xopt_i} \cdot 
	\frac{\delta^2}{n^2}\\
	& \leq  -p(p-1) \sum_s (A_s \xopt)^{p-2} \frac{(\sum_i a_{si})^2}{(Nn)^2}
	\cdot \frac{1}{(Nn)^2} \cdot \frac{\delta^2}{n^2}\\
	& \leq -p(p-1) \cdot \frac{\delta^2}{N^4 n^6} \cdot \norm{A\xopt}{p}^p
\end{align}
Note that we used the facts that entries $a_{ij}$ lie in
$[\frac{1}{N},1]$ and that $\xopt_i \in [\frac{1}{(Nn)^2}, 1]$.
Thus it only remains to bound the third order terms ($g_f$,
in Eqn.\eqref{eq:taylor}). This contribution equals
\begin{equation}
g_f(\eps) = \frac{1}{3!} \sum_i \eps_i^3 \frac{\partial^3 f}{\partial x_i^3} +
 \frac{1}{2!} \sum_{i,j} \eps_i^2 \eps_j \frac{\partial^3 f}{\partial x_i^2 \partial x_j}+
\sum_{i<j<k} \eps_i \eps_j \eps_k \frac{\partial^3 f}{\partial x_i \partial x_j \partial x_k}
\end{equation}
It can be shown by expanding out, and using the facts that $m_{si} \leq N (M_s \xopt)$
and $\frac{1}{\xopt_i} \leq (Nn)^2$, that for $i,j,k$, 
\[ \frac{\partial^3 f}{\partial x_i \partial x_j \partial x_k} \leq 
	10p^3 (Nn)^3 \norm{A\xopt}{p}^p .\]
Thus, the higher order terms can be bounded by
\[g_f(\eps) \leq 10p^3 \cdot n^6 N^3 \cdot \norm{A\xopt}{p}^p \cdot \delta^3 \]
So, if $\delta < \frac{1}{10p^3} \cdot \frac{1}{(Nn)^{12}}$, the
Hessian term dominates. Thus we have, as desired:
\[ f(x) \leq  f(\xopt) \Big( 1-\frac{\delta^2}{(Nn)^6} \Big) \]
\end{proof}

% END PASTE

This proves that the vector we obtain at the end of the $T$
iterations (for $T$ as specified)
has an $\ell_1$ distance at most $\frac{1}{(Nn)^c}$ to $\xopt$. Thus we have
a polynomial time algorithm to compute $\xopt$ to any accuracy.

\section{An Application - $O(\log n)$ Oblivious routing scheme for $\ell_p$} \label{sec:oblivious}
\newcommand{\vect}[1]{\mathbf #1}
We believe that our algorithm for computing the $\norm{A}{\qtop}$ (for non-negative matrices) could find good use
as an optimization tool. For instance, eigenvalue computation is used extensively, not just for partitioning and clustering problems, but
also as a subroutine for solving semi-definite programs \cite{sdp}.
We now give one application of our algorithm and the techniques we developed in section ~\ref{sec:stability} to the case of oblivious routing
in the $\ell_p$ norm. 
 
\paragraph{Oblivious routing.} As outlined in the Introduction, the aim in oblivious routing is, given a graph $G = (V,E)$, to specify how to route a unit flow between every pair of vertices in $V$. Now, given a demand vector 
(demands between pairs of vertices), these unit flows are scaled linearly by the demands, and routed (let us call this the oblivious flow).
This oblivious flow is compared to the best flow in hindsight i.e. knowing the demand vector, with respect to some objective (say congestion), and we need to come up with 
a scheme which bounds this competitive ratio in the worst case.
     
Gupta et al.~\cite{gupta} consider the oblivious routing problem where the cost of a solution is the $\ell_p$ norm of the `flow vector' (the vector consisting of total flow on each edge). In the case $p=\infty$, this is the problem of minimizing congestion, for which the celebrated result of~\cite{racke1} gave an $O(\log n)$ competitive scheme. For the $\ell_1$ version of the problem, the optimal solution (as is easily seen) is to route along shortest paths for each demand pair. The $\ell_p$ version tries to trade-off between these two extremes.

By a clever use of zero sum games, \cite{racke} reduced the problem of showing existence good oblivious routing schemes for any $p$ to the $\ell_\infty$ case. This showed (by a \emph{non-constructive} argument) the existence of an $O(\log n)$ oblivious routing scheme for any $p\ge 1$. They then make their result constructive for $p=2$ (the proof relies heavily on eigenvectors being orthogonal).
Using our algorithm for finding the
$\ell_p$-norm of a matrix and the stability of our maxima (Lemma~\ref{lem:stability}), we make the
result constructive for all $\ell_p$.

\paragraph{Zero-sum game framework of \cite{racke}:}
We first give a brief overview of the non-constructive proof from
\cite{racke}. The worst-case demands for any tree-based oblivious routing scheme can be shown to be those with
non-zero demands only on the edges of the graph. 
The \emph{competitive ratio} of any tree-based oblivious routing scheme 
can then be reduced to a matrix $p$-norm computation: 
if $M$ is a $|E|\times |E|$-dimensional matrix which represents a tree-based oblivious routing scheme which specifies unit flows for each demand across an edge of the graph, the competitive ratio is given by
$\max_{\pnorm{\vect{u}} \leq 1} \pnorm{M\vect{u}}$ where $\vect{u} \in \rr^{|E|}$.
   
To show the existence of an oblivious routing scheme with competitive
ratio $O(\log n)$, \cite{racke} define a continuous two player
zero-sum game. The first player (row player) chooses from the set of
all \emph{tree-based oblivious routing matrices} (of dimension  $|E| \times |E|$).
 The second player's (column player) strategy set is the set of vectors $\vect{u} \in \rr^{|E|}$ with positive entries, and $\norm{u}{p}=1$, and the value of the game is $\norm{M\vect{u}}{p}$. 
With a clever use of min-max duality in zero sum games and the oblivious
routing scheme of \cite{racke1} for congestion ($\ell_\infty$-norm) as
a blackbox, \cite{racke} show the non-constructive existence of an oblivious routing
scheme $M$ which gets a value of $O(\log n)$ for all demand vectors. 

Finding such a (tree-based) oblivious routing scheme requires us to solve this zero-sum game efficiently.
The constructive algorithm from \cite{racke} for $\ell_2$, however crucially uses the
ortho-normality of the eigenspace for $\norm{M}{2}$ computation, to solve the aforementioned zero-sum game.   
First we state without proof a
couple of lemmas from \cite{racke}, which will also feature in our
algorithm.

\begin{lemma} \label{lem:racke:load}
Let $OBL$ be a tree-based oblivious routing scheme given by a $|E| \times {n \choose 2}$ dimensional matrix (non-negative entries) and let its restriction to edges be $OBL' \in
\rr^{|E| \times |E|}$. The competitive ratio of the
oblivious algorithm is at most $\norm{OBL'}{p}$.
\end{lemma}
Henceforth, we shall abuse notation and use $OBL$ to refer to both the
tree-based Oblivious routing matrix and its restriction to edges
interchangeably. Further,  
\begin{lemma} \label{lem:racke:1}
For any given vector $\vect{u} \in \rr^{|E|}$, there exists an
tree-based Oblivious routing scheme (denoted by matrix $OBL$) such
that \[\norm{OBL \cdot \vect{u}}{p} \leq O(\log n) \norm{\vect{u}}{p}\]
\end{lemma}
This lemma shows that for every vector $\vect{u}$, there exists some routing scheme (which could depend on the vector) that is $O(\log n)$ competitive.
We will now show how to compute \emph{one} tree-based routing matrix $OBL$ that works for all vectors i.e. $\norm{OBL}{p} \leq 1$. From Lemma \ref{lem:racke:1}, we know that
for every unit vector $\vect{u}$, there exists an tree-based oblivious
routing matrix such that $\norm{M\cdot\vect{u}}{p} \leq O(\log n)$. We use this
to construct one tree-based oblivious routing matrix $OBL$ that works
for every load vector $\vect{u}$. Note that the set of tree-based
oblivious routing schemes is convex. Before, we show how to construct
the oblivious routing scheme, we present a simple lemma which captures
the continuity of the $p$-norm function.

\begin{lemma} \label{lem:continuity}
Let $f=\frac{\norm{Ax}{p}^p}{\norm{x}{p}^p}$, where $A$ is an $n \times n$ matrix with minimum entry $\frac{1}{N}$ and let $y$ be an $n$-dimensional vector with minimum entry $\frac{1}{(Nn)^2}$. Let $x$ be a vector in the $\delta$-neighborhood of $y$ i.e. $\norm{x-y}{1} = \delta \leq \frac{1}{(Nn)^{12}}$. Then,  
\begin{equation} \label{eq:continuity}
f(x) \leq f(y) + 1 
\end{equation}
\end{lemma}
\begin{proof}
The proof follows just from the continuity and differentiability of the $p$-norm function at every point. Using the Taylor's expansion of $f$, we see that
\begin{equation}\label{eq:taylor2}
f(x)=f(y)+\eps \cdot \grad f_{y} + \frac{1}{2}~ \eps'^T
	H_{f|y} \eps' 
\end{equation} 
where $0 \leq \eps' \leq \eps$. Choosing $\eps=\delta=\frac{1}{(Nn)^{12}}$ and using the lower bounds the matrix entries and the co-ordinates of $y$ as in Lemma ~\ref{lem:stability}, we see that the lemma follows.
\end{proof}
We now sketch how to find a tree-based oblivious routing matrix when the aggregation function is an $\ell_p$ norm. 

\begin{theorem}
There exists a polynomial time algorithm that computes an oblivious routing scheme with competitive ratio $O(\log n)$ when the aggregation function is the $\ell_p$ norm with $p>1$ and the load function on the edges is a norm. 
\end{theorem}  
\begin{proof}[Proof sketch]
The algorithm and proof follow roughly along the lines of the constructive version for $p=2$ in \cite{racke}. As mentioned earlier, their proof uses inner products among the vectors (and the computation of eigenvalues). However, we show that the procedure still works because of the stability of our solution (Lemma~\ref{lem:stability}. 

\newcommand{\poly}{\text{poly}}

Let $J_{\eps}$ be an $|E| \times |E|$ matrix will all entries being $\epsilon$. Let $f(M) = \norm{M+J_{\frac{1}{|E|}}}{p}$. We want a tree-based oblivious routing matrix $OBL$ such that $f(OBL) \leq c \log n$ for some large enough constant $c$. We follow an iterative procedure to obtain this matrix $OBL$ starting with an arbitrary tree-based routing matrix $\Mi{0}$. At stage $i$, we check if for the current matrix $M_i$ , $\norm{M_{i}}{p} \leq c\log n$. If not, using the iterative algorithm in Section \ref{sec:algo}, we obtain  unit vector $x^*_{(i)}$ which maximizes $\norm{M_{(i)}x}{p}$. Let $\tilde{M}_{(i)}$ be the tree-based oblivious routing matrix from Lemma \ref{lem:racke:1} such that $\norm{\tilde{M}_{(i)}x^*_{(i)}}{p} \leq c \log n /2 -2$. We now update \[M_{i+1}= (1-\lambda) M_{i} + \lambda \tilde{M}_{i}\] Observe that this is also a tree-based oblivious routing matrix. We now show that $\norm{\Mi{i+1}}{p}$ decreases by an amount $\Omega(\frac{1}{\poly(n)})$.

At step $i$, roughly speaking, for all vectors $y$ that are far enough from $\xopt_{(i)}$,
 $\norm{\Mi{i} y}{p} \leq \norm{\Mi{i} y}{p} - \frac{1}{\poly(n)}$ from Lemma~\ref{lem:stability} (stability). Choosing $\lambda = \Theta(n^{-c})$ for some large enough constant $c>0$, it easily follows that $\norm{\Mi{i+1} y}{p} \leq \norm{\Mi{i} y}{p} - \frac{1}{\poly(n)}$. On the other hand, consider $y$ in the $\delta$-neighborhood of $x^*_{(i)}$. Using Lemma~\ref{lem:continuity},
\[\norm{\tilde{M}_{i} y}{p} \leq \frac{c \log n}{2} \] 
Hence, 
\begin{eqnarray*}
\norm{\Mi{i+1} y}{p} &=& (1-\lambda) \norm{\Mi{i} y}{p} + \lambda \frac{c}{2} \log n \\ 
&\leq& \norm{\Mi{i} y}{p} - \lambda \times \frac{c}{2}\log n \qquad \mbox{(since } \norm{\Mi{i} y}{p} \geq c\log n \mbox{ )}\\
&\leq& \norm{\Mi{i} y}{p} - \frac{1}{\poly(n)}
\end{eqnarray*}
Hence, it follows that the matrices $M_{i}$ decrease in their $p$-norm by a small quantity $\Omega(\frac{1}{\poly(n)})$ in every step. It follows that this iterative algorithm finds the required tree-based oblivious routing scheme in $\poly(n)$ steps.
\end{proof}

\section{Inapproximability results} \label{sec:lower-bounds}
We will now prove that it is NP-hard to approximate $\norm{A}{\qtop}$-norm of a
matrix to any fixed constant, for any $q \geq p >2$. We then show how this proof carries over to the hardness of computing the $\infty \mapsto p$ norm.

\subsection{Inapproximability of $\norm{A}{\ptop} $}\label{sec:lb:ptop}

Let us start with the question of approximating $\norm{A}{\ptop}$. We first show the following:

\begin{prop} \label{prop:ptop}
For $p\geq 2$ it is NP-hard to approximate that $\pnorm{A}$ to some (small) constant factor $\eta>1$.
\end{prop}
\emph{Proof:}
We give a reduction from the gap version of the MaxCut problem. The following is well-known (c.f.~\cite{hastad})
\begin{quote}
There exist constants $1/2 \leq \rho < \rho' <1$ such that given a regular graph $G=(V,E)$ on $n$ vertices and degree $d$, it is hard to distinguish between:\\
{\sc Yes} case: $G$ has a cut containing at least $\rho' (nd/2)$ edges, and\\
{\sc No} case: No cut in $G$ cuts more that $\rho (nd/2)$ edges.
\end{quote}

Suppose we are given a graph $G = (V, E)$ which is regular and has
degree $d$. The $p$-norm instance we consider will be that of
maximizing $g(x_0, \dots, x_n)$ ($x_i \in \mathbb{R}^n$), defined by
\[ g(x_0, x_1, \dots, x_n) = \frac{ \sum_{i \sim j} |x_i - x_j|^p + Cd
\cdot \big( \sum_i |x_0 + x_i|^p + |x_0 - x_i|^p \big) }{ n |x_0|^p +
\sum_i |x_i|^p }. \]
Here $C$ will be chosen appropriately later. Note that if we divide by $d$, we can see $g(\mathbf{x})$ as the ratio
\begin{equation}\label{eq:g-rewrite}
\frac{g(\mathbf{x})}{d} = \frac{ \sum_{i \sim j} \big( |x_i - x_j|^p + C
(|x_0 + x_i|^p  + |x_0 - x_i|^p + |x_0 + x_j|^p  + |x_0 - x_j|^p)
\big)}{ \sum_{i \sim j} 2|x_0|^p + |x_i|^p + |x_j|^p}.
\end{equation}

The idea is to do the analysis on an edge-by-edge basis. Consider the function
\[ f(x,y) = \frac{ |x-y|^p + C \big( |1+x|^p + |1-x|^p + |1+y|^p +
|1-y|^p \big) }{2 + |x|^p + |y|^p}. \]

\paragraph{Definition.}
A tuple $(x,y)$ is {\em good} if both $|x|$ and $|y|$ lie in the
interval $(1-\eps, 1+\eps)$, and $xy <0$. A technical lemma concerning
$f$ is the following 
\begin{lemma} \label{lem:tech-inequality}
For any $\eps >0$, there is a large enough constant $C$ such
that
\begin{equation}
f(x,y) \leq  \begin{cases} C \cdot 2^{p-1} + \frac{(1+\eps) 2^p}{2+|x|^p+|y|^p}, & \text{if $(x,y)$ is good}\\ C \cdot 2^{p-1} & \text{ otherwise} \end{cases}
\end{equation}
\end{lemma}
We now present the proof of Lemma~\ref{lem:tech-inequality}. We first start with a simpler inequality - note that this is where the condition $p>2$ comes in.
\begin{lemma} \label{lem:brute-inequality}
For all $x \in \mathbb{R}$, we have
\[ \frac{ |1+x|^p + |1-x|^p }{1+|x|^p} \le 2^{p-1}. \]
Further, for any $\eps>0$, there exists a $\delta>0$ such that if $|x| \not\in [1-\eps, 1+\eps]$, then 
\[ \frac{ |1+x|^p + |1-x|^p }{1+|x|^p} \le 2^{p-1} -\del. \]
\end{lemma}
\begin{proof}
We may assume $x>0$. First consider $x>1$. Write $x=1+2\theta$, and thus the first inequality simplifies to
\[ \big[ (1+2\theta)^p - (1+\theta)^p \big] \ge (1+\theta)^p - 1+ 2\theta^p. \]
Now consider \[ I = \int_{x=1}^{1+\theta} \big( (x+\theta)^{p-1} - x^{p-1} \big) dx.\] For each $x$, the function being integrated is $\ge \theta^{p-1}$, since $p>2$ and $x>0$. Thus the integral is at least $\theta^p$. Now evaluating the integral independently and simplifying, we get
\[ (1+2\theta)^p - 2(1+\theta)^p + 1 \ge p \cdot \theta^p,\]
which gives the inequality since $p >2$. Further there is a slack of $(p-2)\theta^p$.
Now suppose $0<x<1$. Writing $x=1-2\theta$ and simplifying similarly, the inequality follows.
Further, since we always have a slack, the second inequality is also easy to see.
\end{proof}

\begin{proof}[Proof (of Lemma~\ref{lem:tech-inequality})]
The proof is a straight-forward case analysis. Call $x$ (resp. $y$) `bad' if $|x| \not\in [1-\eps, 1+\eps]$. Also, $b(x)$ denotes a predicate which is $1$ if $x$ is bad and $0$ otherwise.\\
{\em Case 1.} $(x,y)$ is good. The upper bound in this case is clear (using Lemma~\ref{lem:brute-inequality}).\\
{\em Case 2.} Neither of $x, y$ are bad, but $xy>0$. Using Lemma~\ref{lem:brute-inequality}, we have $f(x,y) \le C\cdot 2^{p-1} + \eps$, which is what we want.\\
{\em Case 3.} At least one of $x, y$ are bad (i.e., one of $b(x), b(y)$ is 1). In this case Lemma~\ref{lem:brute-inequality} gives
\begin{align*}
f(x,y) &\le \frac{|x-y|^p + C \big( (1+|x|^p) (2^{p-1} -\del b(x)) + (1+|y|^p)(2^{p-1} - \del b(y)) \big)}{2+|x|^p + |y|^p} \\
& = C \cdot 2^{p-1} + \frac{|x-y|^p - C\big( \delta b(x) (1+|x|^p) + \delta b(y) (1+|y|^p) \big)}{2+|x|^p +|y|^p}
\end{align*}
Since $|x-y|^p \le 2^{p-1}(|x|^p + |y|^p)$, and one of $b(x), b(y) >0$, we can choose $C$ large enough (depending on $\del$), so that $f(x,y) \le C\cdot 2^{p-1}$.
\end{proof}

\paragraph{Soundness.} Assuming the lemma, let us see why the analysis of the {\sc No} case
follows. Suppose the graph has a Max-Cut value at most $\rho$, i.e., every cut has at most $\rho
\cdot nd/2$ edges. Now consider the vector $x$ which maximizes $g(x_0,
x_1, \dots, x_n)$. It is easy to see that we may assume $x_0 \neq 0$, thus we can scale the vector s.t. $x_0=1$. Let $S \subseteq
V$ denote the set of `good' vertices (i.e., vertices for which $|x_i|
\in (1-\eps, 1+\eps)$).

\begin{lemma}
The number of good edges is at most $\rho \cdot \frac{(|S|+n)d}{4}$.
\end{lemma}
\begin{proof}
Recall that good edges have both end-points in $S$, and further the
corresponding $x$ values have opposite signs. Thus the lemma
essentially says that there is no cut in $S$ with $\rho \cdot
\frac{(|S|+n)d}{4}$ edges.

Suppose there is such a cut. By greedily placing the vertices of $V
\setminus S$ on one of the sides of this cut, we can extend it to a
cut of the entire graph with at least
\[ \rho \cdot \frac{(|S|+n)d}{4} + \frac{(n-|S|)d}{4} = \frac{\rho
nd}{2} + \frac{(1-\rho)(n-|S|)}{4} > \frac{\rho nd}{2} \]
edges, which is a contradiction. This gives the bound.
\end{proof}

Let $\num$ denote the numerator of Eq.\eqref{eq:g-rewrite}. We have
\begin{align*}
\num &= \sum_{i \sim j} f(x_i, x_j) (2+|x_i|^p + |x_j|^p) \\
       &\leq C \cdot 2^{p-1} \cdot \big(nd + d \sum_i |x_i|^p \big) +
\sum_{i \sim j, \text{ good}} (1+\eps)2^p \\
       &\leq Cd \cdot 2^{p-1} \cdot \big(n + \sum_i |x_i|^p \big) +
\frac{\rho d (n+|S|)}{4} \cdot 2^p (1+\eps).
\end{align*}
Now observe that the denominator is $n + \sum_i |x_i|^p \geq n +
|S|(1-\eps)^p$, from the definition of $S$.  Thus we obtain an upper
bound on $g(x)$
\[ g(x) \leq Cd \cdot 2^{p-1} + \frac{\rho d}{4}
 \cdot 2^p (1+\eps)(1-\eps)^{-p}. \]

\paragraph{Hardness factor.} 
In the {\sc Yes} case, there is clearly an assignment of $\pm 1$ to $x_i$ such that $g(\mathbf{x})$ is at least $Cd \cdot 2^{p-1} + \frac{\rho' d}{4} \cdot 2^p$. Thus if $\eps$ is small enough (this will make us pick $C$ which is large), the gap between the optimum values in the {\sc Yes} and {\sc No} cases can be made $\big( 1 + \frac{\Omega(1)}{C} \big)$, where the $\Omega(1)$ term is determined by the difference $\rho' - \rho$.
This proves that the $p$-norm is hard to approximate to some fixed constant factor.

\noindent {\em Note.} In the analysis, $\eps$ was chosen to be a small constant depending on
$p$ and the gap between $\rho$ and $\rho'$; $C$ is a constant chosen
large enough, depending on $\eps$.

\paragraph{The Instance.} We have argued about the hardness of computing the function $g(x_0,x_1,\dots,x_n)$ to some constant factor. This can be formulated as an instance of $p$-norm in a natural way. We describe this formally (though this is clear, the formal description will be useful when arguing about certain properties of the tensored instance which we need for proving hardness of $\norm{A}{\qtop}$ for $p < q$).

First we do a simple change of variable and let $z=n^{1/p}x_0$. Now, we construct the $5|E| \times (n+1)$ matrix $M$. For each edge $e=\{i,j\}$ in $E(G)$, we have five rows in $M$. Let the column indices run from $0 \leq \ell \leq n$.
\begin{equation*}
{\mathbf M_{e_1,\ell}}=  \begin{cases} 1 \text{ if }\ell=i &\text{ and } \qquad -1 \text{ if } \ell=j \\
			0 &\text{ otherwise}
\end{cases}
\end{equation*}
\begin{eqnarray*}
{\mathbf M_{e_2,\ell}}=  \begin{cases} n^{-1/p} \qquad \text{ if }\ell=0  &\text{ and}
			\qquad -1 \qquad \text{ if } \ell=i \\
			0 &\text{ otherwise} 
\end{cases}\\
{\mathbf M_{e_3,\ell}}=  \begin{cases} n^{-1/p} \qquad \text{ if }\ell=0  &\text{ and} 
			\qquad 1 \qquad \text{ if } \ell=i \\
			0 &\text{ otherwise} 
\end{cases}\\
\end{eqnarray*}

We have two similar rows for ${\mathbf M_{e_4,\ell}}$ and ${\mathbf M_{e_5,\ell}}$ 
where we have corresponding values with $j$ instead of $i$. It is easy to see that this 
matrix actually takes the same value $\pnorm{M}$ as $g$. Further in the {\sc Yes} case, 
there is a vector $x=(n^{1/p},x_1,x_2,\dots,x_n)$ with $x_i=\pm 1$, that attains 
the high value ($C.d.2^{p-1}+\rho'd.2^{p-2}$).

\subsection{Amplifying the gap by tensoring}\label{sec:tensoring}
We observe that the matrix $\ptop$-norm is multiplicative under tensoring. More precisely,
\begin{lemma}\label{lem:tensor}
Let $M$, $N$ be square matrices with dimensions $m \times m$ and $n
\times n$ respectively, and let $p \ge 1$. Then $\norm{M \otimes N}{p} = \norm{M}{p} \cdot
\norm{N}{p}$.
\end{lemma}
The tensor product $M \otimes N$ is
defined in the standard way -- we think of it as an $m \times m$
matrix of blocks, with the $i,j$th block being a copy of $N$ scaled by
$m_{ij}$. It is well-known that eigenvalues ($p=2$) mutliply under tensoring. We note that it is crucial that we consider $\norm{A}{p}$. Matrix norms $\norm{A}{\qtop}$ for $p \neq q$ do not in general multiply upon tensoring.
\newcommand{\subvec}[2]{#1^{(#2)}}
\begin{proof}
Let $\lambda(A)$ denote the $p$-norm of a matrix $A$. 
Let us first show the easy direction, that $\lambda(M \otimes N) \geq \lambda(M) \cdot
\lambda(N)$. Suppose $x, y$ are the vectors which `realize' the
$p$-norm for $M, N$ respectively. Then
\begin{align*}
\norm{(M \otimes N)(x \otimes y)}{p}^p &= \sum_{i,j} |(M_i \cdot x)(N_j
\cdot y)|^p\\ &= \Big( \sum_{i}(M_i \cdot x)^p \Big) \Big( \sum_{j}
(N_j \cdot y)^p \Big)\\
&= \lambda(M)^p \cdot \lambda(N)^p
\end{align*}
Also $\norm{x \otimes y}{p} = \norm{x}{p} \cdot \norm{y}{p}$, thus the
inequality follows.

Let us now show the other direction, i.e., $\lambda(M \otimes N) \leq
\lambda(M) \cdot \lambda(N)$. Let $x,z$ be $mn$ dimensional vectors such $z=(A \otimes B) x$. We will think of $x,z$ as being divided into $m$ blocks of size $n$ each. Further by
$\subvec{x}{i}$ (and $\subvec{z}{i}$), we denote the vector in $\mathbb{R}^n$ which is formed by
the $i$th block of $x$ (resp. $z$). 

For ease of notation, let us consider $n \times m$ matrix $X$ with the $i^{th}$ column of $X$ being vector $\subvec{x}{i}$ (and similarly define $Z$). At the expense of abusing notation, let $X_j^t$ refer to the $j^{th}$ row of matrix $X$. 
Also, let the element-wise $p$-norm of matrix $M$ be defined as 
\[|M|_{\odot p}=\big(\sum_{i,j} m_{ij}^p \big)^{1/p}  \label{eq:frobnorm}\]

It is easy to observe that $ Z^t = A X^t BT$. Further, $\norm{z}{p}=|Z|_{\odot p}$. 

We now expand out $Z$ and rearrange the terms to separate out the operations of $B$ and $A$, in order to bound $|Z|_{\odot p}$ using the $p$-norms of $A$ and $B$.
Hence, we have
\begin{eqnarray*}
\norm{z}{p}^p &=& |A X^t B^t|_{\odot p}^p \\
&=&  \sum_{i=1}^m \sum_{j=1}^n \Big( (A_i X_1, A_i X_2,\dots, A_iX_n)B_j^t \Big)^p 
\end{eqnarray*}

But from definition, $ \norm{Mx}{p}^p = \sum_k (B_k x)^p \leq \lambda(M) \norm{x}{p}^p $.

Hence, by applying the operator $p$-norm bound of $B$,
\begin{align*} 
\norm{z}{p}^p &\leq \lambda(B)^p \sum_i \norm{(A_i X_1, A_i X_2,\dots, A_i X_n)}{p}^p \\
& = \lambda(B)^p \sum_{i=1}^m \sum_{k=1}^n (A_i X_k)^p = \lambda(B)^p \sum_{k}\sum_i (A_i X_k)^p \\
& \leq \lambda(B)^p \lambda(A)^p \sum_k \norm{X_k}{p}^p\\
&= \lambda(A)^p \lambda(B)^p |X|_{\odot p}^p 
\end{align*}
Since $|X|_{\odot p}=\norm{x}{p}=1$, we have 
$\norm{z}{p} \leq \lambda(A) \lambda(B)$.

\end{proof}

{\em Note.} We note that the tensoring result is stated for square matrices, while the instances above are rectangular. This is not a problem, because we can pad $0$s to the matrix to make it square without changing the value of the norm.

\paragraph{The tensored matrix and amplification}
Consider any constant $\gamma>0$. We consider the instance of the matrix $M$ obtained in the proof of Proposition \ref{prop:ptop}, and repeatedly tensor it with $M$ $k=\log_{\eta}\gamma$ times to obtain $M'=M^{\otimes k}$. From Lemma~\ref{lem:tensor}, there exist $\tau_C$ and $\tau_S$ where $\tau_C/\tau_S > \gamma$ such that in the {\sc No} case, for every vector $\mathbf{y} \in \rr^{(n+1)k}$, $\pnorm{Ay} \leq \tau_S$. 

Further, in the {\sc Yes} case, there is a vector $y'=(n^{1/p},x_1,x_2,\dots,x_n)^{\otimes k}$ where $x_i =\pm 1$ (for $i= 1,2,\dots n$) such that $\pnorm{M'y'} \geq \tau_C$. 

\emph{Note:}
Our techniques work even when we take the tensor product $\log^c n$ for some constant $c$. Thus we can conclude:
\begin{theorem} \label{thm:phardness}
For any $\gamma>0$ and $p \geq 2$, it is NP-hard to approximate the $p$-norm of a matrix within a factor $\gamma$. Also, it is hard to approximate the matrix $p$-norm to a factor of $\Omega(2^{(\log n)^{1-\eps}})$ for any constant $\eps>0$, unless $\mathsf{NP} \subseteq \mathsf{DTIME}(2^{\text{polylog}(n))})$.
\end{theorem}

\paragraph{Properties of the tensored instance:}
We now establish some structure about the tensored instance, which we will use crucially for the hardness of $\qtop$ norm. Let the entries in vector $y'$ be indexed by $k$-tuple $I=(i_1,i_2,\ldots,i_k)$ where $i_k \in \{0,1,\dots,n)$. It is easy to see that 
\[y'_I=\pm n^{w(I)/p} \text{ where }w(I) \text{ number of } 0s \text{ in tuple }\]
Let us introduce variables $x_I = n^{-w(I)/p}y_I$ where $w(I)=\text{ number of } 0s \text{ in tuple }I$. It is easy to observe that there is a matrix $B$ such that
\begin{equation*}
\frac{\pnorm{M'y}}{\pnorm{y}} = \frac{\pnorm{B \mathbf{x}}}{\sum_I n^{w(I)}|x_I|^p}=g'(x)
\end{equation*}
Further, it can also be seen that in the {\sc Yes} case, there is a $\pm 1$ assignment for $x_I$ which attains the value $g'(x)=\tau_C$. 

\subsection{Approximating $\norm{A}{\qtop}$ when $p \neq q$.}\label{sec:lb:ptoq}

Let us now consider the question of approximating $\norm{A}{\qtop}$. The idea is to use the hardness of approximating $\norm{A}{\ptop}$. We observed in the previous section that the technique of amplifying hardness for computing the $\qtop$-norm by tensoring a (small) constant factor hardness does not work when $q \ne p$. However, we show that we can obtain such amplified label-cover like hardness if the instance has some additional structure. In particular, we show the instances that we obtain from the tensoring the hard instances of $\pnorm{A}$ can be transformed to give such hard instances for $\norm{A}{\qtop}$.

We illustrate the main idea by first showing a (small) constant factor hardness: let us start with the following maximization problem (which is very similar to Eqn.\eqref{eq:g-rewrite})
\begin{equation}\label{eq:g-qtop}
 g(x_0, x_1, \dots, x_n) = \frac{ \Big( \sum_{i \sim j} |x_i - x_j|^p + Cd
\cdot \big( \sum_i |x_0 + x_i|^p + |x_0 - x_i|^p \big) \Big)^{1/p} }{ \big( n |x_0|^q +
\sum_i |x_i|^q \big)^{1/q}}.
\end{equation}
Notice that $x_0$ is now `scaled differently' than in Eq.\eqref{eq:g-rewrite}. This is crucial.
Now, in the {\sc Yes} case, we have
\[ \max_{\mathbf{x}} ~ g(\mathbf{x}) \geq \frac{ \big( \rho'(nd/2) \cdot 2^p + Cnd \cdot 2^p \big)^{1/p}}{ (2n)^{1/q} }. \]
Indeed, there exists a $\pm 1$ solution which has value at least the RHS. Let us write $\num$ for the numerator of Eq.\eqref{eq:g-qtop}. Then
\[ g(\mathbf{x}) = \frac{\num}{\big( n |x_0|^p + \sum_i |x_i|^p \big)^{1/p} } \times \frac{ \big( n |x_0|^p + \sum_i |x_i|^p \big)^{1/p} }{ \big( n |x_0|^q + \sum_i |x_i|^q \big)^{1/q} }. \]
Suppose we started with a {\sc No} instance. The proof of the $q=p$ case implies that the first term in this product is at most (to a $(1+\eps)$ factor)
\[ \frac{ \big( \rho (nd/2) \cdot 2^p + Cnd \cdot 2^p \big)^{1/p}}{ (2n)^{1/p} }. \]
Now, we note that the second term is at most $(2n)^{1/p} / (2n)^{1/q}$. This follows because for any vector $y \in \mathbb{R}^n$, we have $\norm{y}{p}/\norm{y}{q} \le n^{(1/p)- (1/q)}$. We can use this with the $2n$-dimensional vector $(x_0, \dots, x_0, x_1, x_2, \dots, x_n)$ to see the desired claim.

From this it follows that in the {\sc No} case, the optimum is at most (upto an $(1+\eps)$ factor)
\[ \frac{ \big( \rho (nd/2) \cdot 2^p + Cnd \cdot 2^p \big)^{1/p}}{ (2n)^{1/q} }. \]
This proves that there exists an $\alpha>1$ s.t. it is NP-hard to approximate $\norm{A}{\qtop}$ to a factor better than $\alpha$.

A key property we used in the above argument is that in the {\sc Yes} case, there exists a $\pm 1$ solution for the $x_i$ ($i \ge 0$) which has a large value. It turns out that this is the only property we need. More precisely, suppose $A$ is an $n \times n$ matrix, let $\alpha_i$ be positive integers (we will actually use the fact that they are integers, though it is not critical). Now consider the optimization problem $\max_{\mathbf{y} \in \rr^n} g(\mathbf{y})$, with
\begin{equation}\label{eq:opt-problem}
g(\mathbf{y}) = \frac{\norm{A\mathbf{y}}{p}}{(\sum_i \alpha_i |y_i|^p)^{1/p}}
\end{equation}

In the previous section, we established the following claim from the proof of Theorem~\ref{thm:phardness}.
\begin{claim} \label{thm:phard-aux}
For any constant $\gamma>1$, there exist thresholds $\tau_C$ and $\tau_S$ with $\tau_C/\tau_S > \gamma$, such that it is NP-hard to distinguish between:
\begin{quote}
{\sc Yes} case. There exists a $\pm 1$ assignment to $y_i$ in \eqref{eq:opt-problem} with value at least $\tau_C$, and \\
{\sc No} case.~ For all $\mathbf{y} \in \rr^n$, $g(\mathbf{y}) \le \tau_S$.
\end{quote}
\end{claim}
\emph{Proof.} Follows from the structure of Tensor product instance.\\

We can now show that Claim~\ref{thm:phard-aux} implies the desired result.
\begin{theorem} \label{thm:qtop-hard}
It is NP-hard to approximate $\norm{A}{\qtop}$ to any fixed constant $\gamma$ for $q \ge p >2$.
\end{theorem}
\begin{proof}
As in previous proof (Eq.\eqref{eq:g-qtop}), consider the optimization problem $\max_{\mathbf{y} \in \rr^n} h(\mathbf{y})$, with
\begin{equation}\label{eq:opt-problem-qtop}
h(\mathbf{y}) = \frac{\norm{A\mathbf{y}}{p}}{(\sum_i \alpha_i |y_i|^q)^{1/q}}
\end{equation}
By definition,
\begin{equation} \label{eq:h-g-relation}
h(\mathbf{y}) = g(\mathbf{y}) \cdot \frac{(\sum_i \alpha_i |y_i|^p)^{1/p}}{(\sum_i \alpha_i |y_i|^q)^{1/q}}.
\end{equation}
{\bf Completeness.} Consider the value of $h(\mathbf{y})$ for $A, \alpha_i$ in the {\sc Yes} case for Claim~\ref{thm:phard-aux}. Let $\mathbf{y}$ be a $\pm 1$ solution with $g(\mathbf{y}) \ge \tau_C$. Because the $y_i$ are $\pm 1$, it follows that 
\[ h(\mathbf{y}) \ge \tau_C \cdot \big( \sum_i \alpha_i \big)^{(1/p)-(1/q)}. \]
{\bf Soundness.} Now suppose we start with an $A, \alpha_i$ in the {\sc No} case for Claim~\ref{thm:phard-aux}.

First, note that the second term in Eq.\eqref{eq:h-g-relation} is at most $\big( \sum_i \alpha_i \big)^{(1/p)-(1/q)}$. To see this, we note that $\alpha_i$ are positive integers. Thus by considering the vector $(y_1, \dots, y_1, y_2, \dots, y_2, \dots )$, (where $y_i$ is duplicated $\alpha_i$ times), and using $\norm{u}{p} / \norm{u}{q} \le d^{(1/p) - (1/q)}$ for $u \in \rr^d$, we get the desired inequality.

This gives that for all $\mathbf{y} \in \rr^n$,
\[ h(\mathbf{y}) \le g(\mathbf{y}) \cdot \big( \sum_i \alpha_i \big)^{(1/p) - (1/q)} \le \tau_S \cdot \big( \sum_i \alpha_i \big)^{(1/p) - (1/q)}. \]

This proves that we cannot approximate $h(\mathbf{y})$ to a factor better than $\tau_C / \tau_S$, which can be made an arbitrarily large constant by Claim~\ref{thm:phard-aux}. This finishes the proof, because the optimization problem $\max_{\mathbf{y} \in \rr^n} h(\mathbf{y})$ can be formulated as a $\qtop$ norm computation for an appropriate matrix as earlier.
\end{proof}

Note that this hardness instance is not obtained by tensoring the $\qtop$ norm hardness instance. It is instead obtained by considering the $\norm{A}{p}$ hardness instance and transforming it suitably. 
\subsection{Approximating $\norm{A}{\infty \mapsto p}$}
\newcommand{\infp}{\infty \mapsto p}
\newcommand{\zz}{\mathbb{Z}}
The problem of computing the $\infp$ norm of a matrix $A$ turns out to have a very natural and elegant statement in terms of column vectors of the matrix $A$. We first introduce the following problem:
\begin{defn}[Longest Vector Problem]\label{defn:lvp}
Let $\mathbf{v_1}, \mathbf{v_2},\dots,\mathbf{v_n}$ be vectors over $\rr$. The Longest Vector problem asks for the 
\[\max_{\mathbf{x} \in \{-1,1\}^n} \norm{\sum_i x_i \mathbf{v_i}}{p} \]
\end{defn}
Note that this problem differs from the well-studied Shortest Vector Problem \cite{svp} for lattices,which has received a lot of attention in the cryptography community over the last decade \cite{svp-crypto1}. The shortest vector problem asks for {\em minimizing} the same objective in Definition~\ref{defn:lvp} when $x_i \in \mathbb{Z}$.  

We now observe that computing the $\infp$ norm of the matrix is equivalent to finding the length of the longest vector, where the vectors $\mathbf{v_i}$ are the columns of $A$.
\begin{obser}
Computing the $\norm{A}{\infp}$ norm of a matrix is equivalent to computing the length of the \emph{Longest vector problem} where the vectors are the column vectors of $A$. 
\end{obser}
\begin{proof}
First note that $\pnorm{A\mathbf{x}}=\pnorm{\sum_i x_i \mathbf{a_i}}$. The observation follows by noticing that this is maximized when $|x_i|=1$ for all $i$. 
\end{proof}

The $\infp$ norm of the matrix also seems like a natural extension of the Grothendieck problem\cite{an,naor}. When $p=1$, we obtain the original Grothendieck problem, and the $p=2$ case is the $\ell_2$ Grothendieck problem and maximizes the quadratic form for p.s.d. matrices. Further, as mentioned earlier there is a constant factor approximation for $1\leq p \leq 2$ using \cite{nesterov}. However, for the $p>2$, we show that there is $\Omega(2^{(\log n)^{1-\epsilon}})$ hardness for computing $\infp$ norm assuming $NP$ does not have quasipolynomial time algorithms using the same techniques from Theorem~\ref{thm:qtop-hard}.
\begin{theorem} \label{thm:infp-hard}
It is NP-hard to approximate $\norm{A}{\infp}$ to any constant $\gamma$ for $p >2$ and hard to approximate within a factor of $\Omega(2^{(\log n)^{1-\eps}})$ for any constant $\eps>0$, assuming $\mathsf{NP} \notin \mathsf{DTIME}(2^{\text{polylog}(n)})$.
\end{theorem}
The proof of Theorem ~\ref{thm:qtop-hard} also works out for $q=\infty$ by noting that
the second expression in Eq.~\eqref{eq:h-g-relation} is instead 
$\max_{\mathbf{x}} \frac{(\sum_i \alpha_i |x_i|^p)^{1/p}}{||x||_\infty}$ which is also maximized when $|x_i|=1$ for all $i$.

\section{Acknowledgements}
We would like to thank our advisor Moses Charikar for many useful suggestions and comments throughout the progress of the work. We would also like to thank David Steurer for several discussions and pointing out connections to other problems. Finally we would like to thank Rajsekar Manokaran for various interactions about the inapproximability results.   
\bibliographystyle{alpha}
\bibliography{norm}

\appendix
\section{Miscellany}

\subsection{Non-convex optimization}\label{sec:app:non-convex}
Note that computing the $\ptop$ norm is in general not a convex optimization problem. i.e., the function $f$ defined by $f(x) = \frac{\norm{Ax}{p}}{\norm{x}{p}}$ is not in general concave. For example, consider
\[ A =
\left( {\begin{array}{cc}
 1 & 2  \\
 3 & 1  \\
 \end{array} } \right); ~~
\mathbf{x} = \left( {\begin{array}{c} 0.1 \\ 0.1 \\ \end{array} } \right);~~
\mathbf{y} = \left( {\begin{array}{c} 0.2 \\ 0.5 \\ \end{array} } \right)
\]
In this case, with $p=2.5$, for instance, it is easy to check that $f((\mathbf{x} + \mathbf{y})/2) < (f(\mathbf{x}) + f(\mathbf{y})) /2$. Thus $f$ is not concave.
However, it could still be that $f$ raised to a certain power is concave.

\subsection{Duality}\label{sec:app:lpspace-duality}
The following equality is useful in `moving' from one range of parameters to another. We use the fact that $\norm{u}{p} = \max_{y~:~\norm{y}{p'}=1} y^Tx$, where $p'$ is the `dual norm', satisfying $1/p + 1/p' = 1$. (similarly $q'$ denotes the dual norm of $q$)
\begin{equation}\label{eq:duality-def}
 \norm{A}{\qtop} = \max_{\norm{x}{q}=1} \norm{Ax}{p} = \max_{\substack{\norm{x}{q} =1 \\ \norm{y}{p'}=1}} y^T Ax = \max_{\substack{\norm{x}{q} =1 \\ \norm{y}{p'}=1}} x^T A^T y = \norm{A^T}{p' \mapsto q'}
\end{equation}

\subsection{Moving to a positive matrix}\label{sec:app:simplification}
We now show that by adding a very small positive number to each entry of the matrix, the $\qtop$-norm does not change much.
\begin{lemma}
Let $A$ be an $n \times n$ matrix where the maximum entry is scaled to $1$. Let $J_\epsilon$ be the matrix with all entries being $\epsilon$.
\[\norm{A+J_{\eps}}{\qtop} \leq \norm{A}{\qtop}\big(1+\epsilon n^{1+\frac{1}{p}-\frac{1}{q}}\big)\]
\end{lemma}
\begin{proof}
We first note that $\norm{A}{\qtop} \geq 1$ (because the maximum entry is $1$). It is also easy to see that $J_\eps$ is maximized by the vector with all equal entries.
Hence $\norm{J_\eps}{\qtop} \leq n^{1+\frac{1}{p} -\frac{1}{q}} \epsilon$. Hence, by using the fact that $\norm{\cdot}{\qtop}$ is a norm, the lemma follows.
\end{proof}

\end{document}